\documentclass[a4paper, 12pt]{article}
\usepackage{amssymb}
\usepackage{amsmath}
\usepackage{amsthm}

\newcommand{\be}{\begin{equation}}
\newcommand{\ee}{\end{equation}}

\newcommand{\bN}{\mathbb{N}}
\newcommand{\bR}{\mathbb{R}}
\newcommand{\bC}{\mathbb{C}}

\newcommand{\cA}{\mathcal{A}}
\newcommand{\cC}{\mathcal{C}}
\newcommand{\cF}{\mathcal{F}}
\newcommand{\cH}{\mathcal{H}}
\newcommand{\cK}{\mathcal{K}}
\newcommand{\cM}{\mathcal{M}}
\newcommand{\cO}{\mathcal{O}}
\newcommand{\cP}{\mathcal{P}}
\newcommand{\cR}{\mathcal{R}}
\newcommand{\cS}{\mathcal{S}}
\newcommand{\cT}{\mathcal{T}}
\newcommand{\cV}{\mathcal{V}}
\newcommand{\cW}{\mathcal{W}}

\newcommand{\mfc}{\mathfrak{c}}
\newcommand{\mfu}{\mathfrak{u}}
\newcommand{\mfg}{\mathfrak{g}}
\newcommand{\mfgl}{\mathfrak{gl}}
\newcommand{\mfk}{\mathfrak{k}}
\newcommand{\mfL}{\mathfrak{L}}

\newcommand{\bsone}{\boldsymbol{1}}
\newcommand{\bsq}{\boldsymbol{q}}
\newcommand{\bsp}{\boldsymbol{p}}
\newcommand{\bslambda}{\boldsymbol{\lambda}}
\newcommand{\bsLambda}{\boldsymbol{\Lambda}}
\newcommand{\bstheta}{\boldsymbol{\theta}}
\newcommand{\bsC}{\boldsymbol{C}}
\newcommand{\bsQ}{\boldsymbol{Q}}

\newcommand{\ri}{\text{i}}
\newcommand{\dd}{\text{d}}
\newcommand{\diag}{\text{diag}}
\newcommand{\tr}{\text{tr}}
\newcommand{\bc}{\text{BC}}
\newcommand{\ext}{\text{ext}}
\newcommand{\half}{\frac{1}{2}}

\theoremstyle{plain}
\newtheorem{THEOREM}{Theorem}
\newtheorem{LEMMA}[THEOREM]{Lemma}

\begin{document}

\begin{center}
    \Large{\textbf{
    Scattering theory of the hyperbolic $BC_n$ Sutherland and the rational 
    $BC_n$ Ruijsenaars--Schneider--van Diejen models
    }}
\end{center}

\bigskip
\begin{center}
    B.G.~Pusztai\\
    Bolyai Institute, University of Szeged,\\
    Aradi v\'ertan\'uk tere 1, H-6720 Szeged, 
    Hungary\\
    e-mail: \texttt{gpusztai@math.u-szeged.hu}
\end{center}

\bigskip
\begin{abstract}
In this paper, we investigate the scattering properties of the hyperbolic 
$BC_n$ Sutherland and the rational $BC_n$ Ruijsenaars--Schneider--van Diejen 
many-particle systems with three independent coupling constants. Utilizing
the recently established action-angle duality between these classical 
integrable models, we construct their wave and scattering maps. In 
particular, we prove that for both particle systems the scattering map
has a factorized form.



\end{abstract}
\newpage

\section{Introduction}
\label{SECTION_Introduction}
\setcounter{equation}{0}
The Calogero--Moser--Sutherland (CMS) many-particle systems (see e.g. 
\cite{Calogero}, \cite{Sutherland}, \cite{Moser_1975}, 
\cite{Olsha_Pere_76}) and their relativistic deformations, the 
Ruijsenaars--Schneider--van Diejen (RSvD) models (see e.g. 
\cite{Ruij_Schneider}, \cite{van_Diejen}) are among the most actively 
studied integrable systems. They appear in several branches of mathematics 
and physics, with numerous applications ranging from symplectic geometry, 
Lie theory and harmonic analysis to solid state physics and Yang--Mills 
theory. The intimate connection with the theory of soliton equations is a 
particularly important and appealing feature of these finite dimensional 
integrable systems. It is a remarkable fact that the CMS and the RSvD models 
associated with the $A_n$ root system can be used to describe the soliton 
interactions of certain integrable field theories defined on the whole real 
line (see e.g. \cite{Ruij_Schneider}, \cite{Ruij_CMP1988}, 
\cite{Ruij_RIMS_2}, \cite{Ruij_RIMS_3}, \cite{Babelon_Bernard}). In 
particular, these particle systems are characterized by conserved asymptotic 
momenta and factorized scattering maps (see e.g. \cite{Ruij_CMP1988},
\cite{Kulish_1976}, \cite{Moser_1977},  
\cite{Ruij_FiniteDimSolitonSystems}). That is, in perfect analogy with the 
behavior of the solitons, the $n$-particle scattering is completely 
determined by the $2$-particle processes. However, apart form some heuristic 
arguments \cite{Kapustin_Skorik}, the link between the integrable boundary 
field theories and the particle models associated with the non-$A_n$-type 
root systems has not been developed yet. In our paper \cite{Pusztai_JPA2011} 
we have proved that the classical hyperbolic $C_n$ Sutherland model also has 
a factorized scattering map, but to our knowledge analogous results for the 
other non-$A_n$-type root systems are not available. Motivated by this fact, 
in this paper we work out in detail the scattering theory of the hyperbolic 
$BC_n$ Sutherland and the rational $BC_n$ RSvD models with the maximal
number of independent coupling constants.

Upon introducing the subset
\be
    \mfc 
    = \{ x = (x_1, \ldots, x_n) \in \bR^n 
        \, | \, 
        x_1 > \ldots > x_n > 0 \} \subset \bR^n,
\label{mfc}
\ee
let us recall that the classical hyperbolic $BC_n$ Sutherland and the 
rational $BC_n$ RSvD models live on the phase spaces
\be
    \cP^S = \{ (q, p) \, | \, q \in \mfc, \, p \in \bR^n \}
    \quad \text{and} \quad
    \cP^R = 
        \{ (\lambda, \theta) 
            \, | \, 
            \lambda \in \mfc, \, \theta \in \bR^n \},
\label{cP_S&R}
\ee
respectively. By endowing these spaces with the natural symplectic forms
\be 
    \omega^S = \sum_{a = 1}^n \dd q_a \wedge \dd p_a 
    \quad \text{and} \quad 
    \omega^R = \sum_{a = 1}^n \dd \lambda_a \wedge \dd \theta_a,
\label{omega_S&R}
\ee
we may think of $\cP^S$ and $\cP^R$ as two different copies of the cotangent 
bundle $T^* \mfc$. The hyperbolic $BC_n$ Sutherland model is characterized 
by the interacting many-body Hamiltonian
\be
    H^S =  
    \sum_{a = 1}^n 
        \left( 
            \frac{p_a^2}{2} + g_1^2 w(q_a) + g_2^2 w(2 q_a) 
        \right) 
    + \sum_{1 \leq a < b \leq n}
        \left( 
            g^2 w(q_a -q_b) + g^2 w(q_a + q_b)
        \right), 
\label{H_S}
\ee
with potential function $w(x) = \sinh(x)^{-2}$. To ensure the repulsive 
nature of the interaction, on the real coupling parameters $g$, $g_1$ and 
$g_2$ we impose the constraints $g^2 > 0$ and $g_1^2 + g_2^2 > 0$. As for 
the rational $BC_n$ RSvD model, the dynamics is governed by the Hamiltonian
\be
    H^R 
    = \sum_{a = 1}^n \cosh(2 \theta_a) v_a(\lambda)
        + \frac{\nu \kappa}{4 \mu^2}
            \prod_{a = 1}^n 
            \left(
                1 + \frac{4 \mu^2}{\lambda_a^2}
            \right)
        - \frac{\nu \kappa}{4 \mu^2},
\label{H_R}
\ee
where the real parameters $\mu$, $\nu$ and $\kappa$ appear in the function
\be
    v_a(\lambda) =
    \left(
        1 + \frac{\nu^2}{\lambda_a^2}
    \right)^\half 
    \left(
        1 + \frac{\kappa^2}{\lambda_a^2}
    \right)^\half
    \prod_{\substack{b = 1 \\ (b \neq a)}}^n
    \left(
        1 + \frac{4 \mu^2}{(\lambda_a - \lambda_b)^2}
    \right)^\half
    \left(
        1 + \frac{4 \mu^2}{(\lambda_a + \lambda_b)^2}
    \right)^\half,
\label{v_a}
\ee
too. Let us note that on the RSvD coupling parameters we impose
the conditions $\mu \neq 0$, $\nu \neq 0$ and $\nu \kappa \geq 0$. 

Working in a symplectic reduction framework, in our paper 
\cite{Pusztai_NPB2012} we established the action-angle duality between the 
hyperbolic $BC_n$ Sutherland and the rational $BC_n$ RSvD models, provided 
the coupling parameters satisfy the relations
\be
	g^2 = \mu^2, 
	\quad 
	g_1^2 = \half \nu \kappa, 
	\quad 
	g_2^2 = \half (\nu - \kappa)^2.
\label{coupling_relations}
\ee
By continuing our work \cite{Pusztai_NPB2012}, in this paper we explore the 
scattering theory of these particle systems. More precisely, after a brief 
overview in section \ref{SECTION_Preliminaries} on the necessary background 
material, in section \ref{SECTION_Scattering_theory} we examine the temporal 
asymptotics of the trajectories of the Sutherland and the RSvD many-body 
models. The outcome of our analysis is formulated in lemmas 
\ref{LEMMA_Sutherland_asymptotics} and \ref{LEMMA_RSvD_asymptotics}. 
As expected, the asymptotics naturally lead to the wave and the scattering 
maps, too. Quite remarkably, the symplecticity of the wave maps can be seen 
as an immediate consequence of the duality between these integrable many-body 
models. We also prove that the scattering maps have factorized forms, which 
is the characteristic property of the soliton systems. Our main results on 
the wave and the scattering maps are summarized in theorems 
\ref{THEOREM_Sutherland_scattering} and \ref{THEOREM_RSvD_scattering}. To 
conclude the paper, in section \ref{SECTION_Discussion} we discuss our
results and offer some open problems related to the scattering theory of 
the CMS and the RSvD systems.

\section{Preliminaries}
\label{SECTION_Preliminaries}
\setcounter{equation}{0}
One of the main ingredients of the scattering theory we wish to develop in 
this paper is the action-angle duality between the hyperbolic $BC_n$ Sutherland 
and the rational $BC_n$ RSvD systems \cite{Pusztai_NPB2012}. To keep the paper 
self-contained, in this section we gather some relevant facts about the 
symplectic reduction derivation of these models with emphasis on their duality 
properties. Up to some minor changes in the conventions, our overview closely 
follows \cite{Pusztai_NPB2012}. For convenience, we formulate our results 
using the three independent RSvD coupling parameters $\mu$, $\nu$ and 
$\kappa$. Without loss of generality, we may assume from the outset that 
$\nu > 0 > \mu$ and $\kappa \geq 0$.

\subsection{Background material from symplectic geometry}
Upon taking an arbitrary positive integer $n \in \bN$, we introduce the set 
\be
    \bN_n = \{1, \ldots, n\} \subset \bN.
\label{bN_n}
\ee 
In the following we also frequently use the notation $N = 2 n$. Now, with the
aid of the unitary $N \times N$ matrix 
\be
    \bsC 
    = \begin{bmatrix}
    	0_n & \bsone_n \\
    	\bsone_n & 0_n
    \end{bmatrix},
\label{bsC} 
\ee
we define the non-compact real reductive matrix Lie group  
\be
	G = U(n, n) = \{ y \in GL(N, \bC) \, | \, y^* \bsC y = \bsC \}
\label{G}
\ee
with Lie algebra
\be
	\mfg 
	= \mfu(u, n) 
	= \{ Y \in \mfgl(N, \bC) \, | \, Y^* \bsC + \bsC Y = 0 \}.
\label{mfg}
\ee
Turning to $\cP = G \times \mfg$, at each point $(y, Y) \in \cP$ we 
define the $1$-form
\be
    \vartheta_{(y, Y)} (\delta y \oplus \delta Y) 
    = \tr \left( Y y^{-1} \delta y \right)
    \qquad (\delta y \in T_y G, \, \delta Y \in T_Y \mfg \cong \mfg).
\label{vartheta}
\ee
It is clear that the product manifold $\cP$ equipped with the symplectic 
form $\omega = -\dd \vartheta$ provides a convenient model for the 
cotangent bundle $T^* G$.

Next, for each column vector $V \in \bC^N$ we define the $N \times N$ matrix
\be
    \xi(V) = \ri \mu (V V^* - \bsone_N) + \ri (\mu - \nu) \bsC.
\label{xi}
\ee
Also, by taking the unitary elements, in $G$ (\ref{G}) we choose the maximal 
compact subgroup 
\be
    K = U(n, n) \cap U(N) \cong U(n) \times U(n).
\label{K}
\ee 
In its Lie algebra 
$\mfk = \mfu(n, n) \cap \mfu(N)$ the subset
\be
    \cO = 
    \{ \xi(V)
        \, | \,
        V \in \bC^N, \, V^* V = N, \, \bsC V + V = 0 \} \subset \mfk
\label{cO}
\ee
forms an orbit under the adjoint action of $K$. Consequently, it comes 
naturally equipped with the Kirillov--Kostant--Souriau symplectic structure 
$\omega^\cO$ having the form
\be
    \omega^\cO_\rho( [X, \rho], [Z, \rho] ) 
    = \tr \left( \rho [X, Z] \right),
\label{omega_cO}
\ee
where $\rho \in \cO$ and $[X, \rho], [Z, \rho] \in T_\rho \cO$
with some $X, Z \in \mfk$.

Our study of the CMS and the RSvD particle systems is based on the symplectic 
reduction of the extended phase space $\cP^\ext = \cP \times \cO$ endowed 
with the symplectic form
\be
    \omega^\ext = \half (\omega + \omega^\cO).
\label{omega_ext}
\ee 
In passing we mention that the factor $1/2$ is inserted into this definition 
purely for convenience. Now, let us note that the symplectic left action of 
the product Lie group $K \times K$ on $\cP^\ext$ defined by the formula
\be
    (k_L, k_R) \cdot (y, Y, \rho)
    = (k_L y k_R^{-1}, k_R Y k_R^{-1}, k_L \rho k_L^{-1})
    \qquad
    ((k_L, k_R) \in K \times K)
\label{KK_action}
\ee
admits a $K \times K$-equivariant momentum map. Making use of the 
identification $(\mfk \oplus \mfk)^* \cong \mfk \oplus \mfk$ induced 
by an appropriate multiple of the $\tr$ functional, the momentum map can
be written as
\be
	J^\ext \colon \cP^\ext \rightarrow \mfk \oplus \mfk,
	\quad
	(y, Y, \rho)
	\mapsto
	( (y Y y^{-1})_\mfk + \rho ) \oplus (- Y_\mfk - \ri \kappa \bsC),
\label{J_ext}
\ee
where $Y_\mfk$ denotes the anti-Hermitian part of the matrix $Y$. 

The point of the above discussion is that by starting with certain 
$K \times K$-invariant Hamiltonians on $\cP^\ext$, both the hyperbolic 
$BC_n$ Sutherland and the rational $BC_n$ RSvD models with three independent 
coupling constants can be derived by reducing the symplectic manifold 
$(\cP^\ext, \omega^\ext)$ at the zero value of the momentum map $J^\ext$. 

\subsection{The Sutherland model from Marsden--Weinstein reduction}
One of the key objects in the symplectic reduction derivation of the 
hyperbolic $BC_n$ Sutherland model is its Lax matrix 
$L \colon \cP^S \rightarrow \mfg$. It has the form
\be
	L
    = \begin{bmatrix}
	 	A & B \\
	 	-B & -A
	\end{bmatrix}
	- \ri \kappa \bsC,
\label{L}
\ee
where $A$ and $B$ are appropriate $n \times n$ matrices. More precisely, 
their entries are
\be
	A_{a, b} = \frac{-\ri \mu}{\sinh(q_a - q_b)}, 
	\quad
	A_{c, c} = p_c, 
	\quad
	B_{a, b} = \frac{\ri \mu}{\sinh(q_a + q_b)}, 
	\quad
	B_{c, c} 
	= \ri \frac{\nu + \kappa \cosh(2 q_c)}{\sinh(2 q_c)}, 
\label{A_and_B}
\ee
where $a, b, c \in \bN_n$ and $a \neq b$. Next, to each real $n$-tuple 
$q = (q_1, \ldots, q_n)$ we associate the matrix
\be
    Q = \diag(q_1, \ldots, q_n, -q_1, \ldots, -q_n),
\label{Q}
\ee
and we also introduce the column vector $E \in \bC^N$ with components 
\be
    E_a = - E_{n + a} = 1 \qquad (a \in \bN_n).
\label{E}
\ee    
Finally, let $U(1)_*$ denote the diagonal embedding of $U(1)$ into 
$K \times K$ and define the product manifold 
\be
    \cM^S = \cP^S \times (K \times K) / U(1)_*.
\label{cM_S} 
\ee

Having equipped with the above objects, we can start the reduction 
procedure by analyzing the level set
\be
    \mfL_0 
    = \{ (y, Y, \rho) \in \cP^\ext \, | \, J^\ext(y, Y, \rho) = 0\}.
\label{mfL_0}
\ee
It turns out to be an embedded submanifold of $\cP^\ext$, and the 
diffeomorphism
\be
	\Upsilon^S_0 \colon \cM^S \rightarrow \mfL_0,
	\quad
	(q, p, (\eta_L, \eta_R) U(1)_*)
	\mapsto
	(\eta_L e^Q \eta_R^{-1},
	    \eta_R L(q, p) \eta_R^{-1}, 
	    \eta_L \xi(E) \eta_L^{-1})
\label{Upsilon_S_0}
\ee
gives rise to the identification $\cM^S \cong \mfL_0$. By inspecting the 
(residual) $K \times K$-action (\ref{KK_action}) on  $\cM^S$, it is clear
that the base manifold of the trivial principal $(K \times K)/U(1)_*$-bundle
\be
	\pi^S \colon \cM^S \twoheadrightarrow \cP^S,
	\quad
	(q, p, (\eta_L, \eta_R) U(1)_*)
	\mapsto (q, p)
\label{pi_S}
\ee
provides a realization of the reduced manifold. Since 
$(\pi^S)^* \omega^S = (\Upsilon_0^S)^* \omega^\ext$, we conclude that the 
Sutherland phase space $\cP^S$ (\ref{cP_S&R}) does serve as a convenient 
model for the reduced symplectic manifold $\cP^\ext / \! /_0 (K \times K)$. 
Moreover, making use of the $K \times K$-invariant Hamiltonian
\be
    F \colon \cP^\ext \rightarrow \bR,
    \quad
	(y, Y, \rho)
	\mapsto
	F	(y, Y, \rho) = \frac{1}{4} \tr(Y^2),
\label{F}
\ee
we find the relationship $(\pi^S)^* H^S = (\Upsilon_0^S)^* F$ as well. 
In other words, the reduced Hamiltonian induced by the quadratic function 
$F$ coincides with the Hamiltonian (\ref{H_S}) of the hyperbolic $BC_n$ 
Sutherland model with coupling constants displayed in equation 
(\ref{coupling_relations}).

\subsection{The symplectic reduction derivation of the RSvD model}
Compared to the Sutherland model, the reduction derivation of the rational 
$BC_n$ RSvD model requires a slightly longer preparation. As a first step, 
for each $a \in \bN_n$ we define the function
\be
    \mfc \ni \lambda
    \mapsto
    z_a(\lambda) 
        = - \left(1 + \frac{\ri \nu}{\lambda_a} \right)
            \prod_{\substack{b = 1 \\ (b \neq a)}}^n
            \left( 1 + \frac{2 \ri \mu}{\lambda_a - \lambda_b} \right)
            \left( 1 + \frac{2 \ri \mu}{\lambda_a + \lambda_b} \right) 
    \in \bC.
\label{z}
\ee
Also, we introduce the function $\cA \colon \cP^R \rightarrow G$ with
matrix entries
\be
\begin{split}
    & \cA_{a, b} 
        = e^{-\theta_a - \theta_b}
            \vert z_a z_b \vert^\half 
            \frac{2 \ri \mu}{2 \ri \mu + \lambda_a - \lambda_b},
    \quad 
    \cA_{n + a, n + b} 
        = e^{\theta_a + \theta_b}
            \frac{ \overline{z}_a z_b}
                {\vert z_a z_b \vert^\half} 
            \frac{2 \ri \mu}{2 \ri \mu - \lambda_a + \lambda_b},
    \\
    & \cA_{a, n + b} 
        = \overline{\cA}_{n + b, a} 
        = e^{- \theta_a + \theta_b} z_b  
            \vert z_a z_b^{-1} \vert^\half
            \frac{2 \ri \mu}{2 \ri \mu + \lambda_a + \lambda_b}
            + \frac{\ri(\mu - \nu)}{\ri \mu + \lambda_a} \delta_{a, b},
\label{cA}
\end{split}
\ee
where $a, b \in \bN_n$. As we have shown in our earlier paper 
\cite{Pusztai_NPB2011}, the positive definite $N \times N$ matrix $\cA$ 
provides a Lax matrix for the rational $C_n$ RSvD model with parameters 
$\mu$ and $\nu$.

Next, with the aid of the $\kappa$-dependent functions
\be
    \alpha(x) = \frac{\sqrt{x + \sqrt{x^2 + \kappa^2}}}{\sqrt{2 x}}
    \quad \mbox{and} \quad
    \beta(x) = \ri \kappa \frac{1}{\sqrt{2 x}}
        \frac{1}{\sqrt{x + \sqrt{x^2 +\kappa^2}}}
\label{alpha&beta}
\ee
defined for $x > 0$, for each $\lambda \in \mfc$ we introduce the 
$N \times N$ matrix
\be
    h(\lambda) 
        = \begin{bmatrix}
        \diag(\alpha(\lambda_1), \ldots, \alpha(\lambda_n)) 
            & \diag(\beta(\lambda_1), \ldots, \beta(\lambda_n)) 
        \\
        -\diag(\beta(\lambda_1), \ldots, \beta(\lambda_n)) 
            & \diag(\alpha(\lambda_1), \ldots, \alpha(\lambda_n))
        \end{bmatrix}
    \in G.
\label{h}
\ee
Utilizing $h(\lambda)$, the Lax matrix of the rational $BC_n$ RSvD model 
can be written as
\be
    \cA^\bc \colon \cP^R \rightarrow G, 
    \quad
    (\lambda, \theta)
    \mapsto
    \cA^\bc(\lambda, \theta) 
        = h(\lambda)^{-1} \cA(\lambda, \theta) h(\lambda)^{-1}.
\label{A_bc}
\ee
Heading toward our goal, we still need the column vector 
$\cF(\lambda, \theta) \in \bC^N$ with components 
\be 
    \cF_a(\lambda, \theta) 
    = e^{-\theta_a} \vert z_a(\lambda) \vert^\half
    \quad \mbox{and} \quad
    \cF_{n + a}(\lambda, \theta) 
    = e^{\theta_a} \overline{z_a(\lambda)} 
    \vert z_a(\lambda) \vert^{-\half},
\label{cF}
\ee
where $a \in \bN_n$, together with the column vector
\be
    \cV(\lambda, \theta) 
    = \cA(\lambda, \theta)^{-\half} 
    \cF(\lambda, \theta)
    \in \bC^N.
\label{cV}
\ee
Also, for each $\lambda = (\lambda_1, \ldots, \lambda_n) \in \bR^n$ 
we define the diagonal matrix
\be
    \Lambda = \diag(\lambda_1, \ldots, \lambda_n,
        -\lambda_1, \ldots, -\lambda_n). 
\label{Lambda}
\ee

At this point we can put the RSvD particle system into the context of 
symplectic reduction. Upon introducing the smooth product manifold 
\be
    \cM^R = \cP^R \times (K \times K) / U(1)_*,
\label{cM_R}
\ee    
the diffeomorphism $\Upsilon_0^R \colon \cM^R \rightarrow \mfL_0$ 
defined by the formula
\be
    (\lambda, \theta, (\eta_L, \eta_R) U(1)_*)
    \mapsto
    (\eta_L \cA(\lambda, \theta)^\half h(\lambda)^{-1} \eta_R^{-1},
        \eta_R h(\lambda) \Lambda h(\lambda)^{-1} \eta_R^{-1},
        \eta_L \xi(\cV(\lambda, \theta)) \eta_L^{-1})
\label{Upsilon_R_0}
\ee
provides an alternative parametrization of the level set $\mfL_0$ 
(\ref{mfL_0}). Moreover, the explicit form of the (residual) 
$K \times K$-action  (\ref{KK_action}) on the model space 
$\cM^R \cong \mfL_0$ permits us to identify the reduced manifold with 
the base manifold of the trivial principal $(K \times K) / U(1)_*$-bundle
\be
    \pi^R \colon \cM^R \twoheadrightarrow \cP^R,
    \quad
    (\lambda, \theta, (\eta_L, \eta_R)U(1)_*) 
    \mapsto 
    (\lambda, \theta).
\label{pi_R}
\ee
Since $(\pi^R)^* \omega^R = (\Upsilon_0^R)^* \omega^\ext$, we conclude 
that the reduced space $\cP^\ext / \! /_0 (K \times K)$ can be naturally 
identified with the RSvD phase space $\cP^R$ (\ref{cP_S&R}). Finally, let 
us consider the function
\be
    f \colon \cP^\ext \rightarrow \bR,
    \quad
    (y, Y, \rho) \mapsto f(y, Y, \rho) = \half \tr(y y^*).
\label{f}
\ee
From the relationship $(\pi^R)^* H^R = (\Upsilon_0^R)^* f$ we infer 
immediately that the reduced Hamiltonian function corresponding to the 
$K \times K$-invariant function $f$ coincides with the Hamiltonian 
(\ref{H_R}) of the rational $BC_n$ RSvD model.

\subsection{Action-angle duality}
The two different parameterizations $\Upsilon^S_0$ (\ref{Upsilon_S_0}) and 
$\Upsilon^R_0$ (\ref{Upsilon_R_0}) of the level set $\mfL_0$ (\ref{mfL_0}) 
induce two different models, $\cP^S$ and $\cP^R$ (\ref{cP_S&R}), of the 
symplectic quotient $\cP^\ext / \! /_0 (K \times K)$. Therefore there is 
a \emph{symplectomorphism}
\be
    \cS \colon \cP^S \rightarrow \cP^R
\label{cS}
\ee 
between $\cP^S$ and $\cP^R$, uniquely characterized by the equation
\be
    \cS \circ \pi^S \circ (\Upsilon^S_0)^{-1} 
        = \pi^R \circ (\Upsilon^R_0)^{-1}.
\label{cS_def}
\ee
Making use of this purely geometric observation we can easily construct 
canonical action-angle variables for both the Sutherland and the RSvD models.
Indeed, by pulling back the canonical positions and momenta of the RSvD 
system, the global coordinates $\cS^* \lambda_a$ and $\cS^* \theta_a$ on 
$\cP^S$ give rise to canonical action-angle variables for the Sutherland 
model. Similarly, by pulling back the canonical positions and momenta of 
the Sutherland system, the global coordinates $(\cS^{-1})^* q_a$ and 
$(\cS^{-1})^* p_a$ on $\cP^R$ provide canonical action-angle variables 
for the RSvD model. This remarkable phenomenon goes under the name of 
action-angle duality between the Sutherland and the RSvD particle systems.

\section{Scattering theory}
\label{SECTION_Scattering_theory}
\setcounter{equation}{0}
In this section we provide a thorough analysis of the time evolution of the 
hyperbolic $BC_n$ Sutherland and the rational $BC_n$ RSvD dynamics for the 
large positive and negative values of time $t$. Based on the resulting 
temporal asymptotics we construct the wave and the scattering maps as well.  

\subsection{Scattering properties of the Sutherland model}
Take an arbitrary point $(q, p) \in \cP^S$ and consider the Hamiltonian flow 
of the Sutherland model
\be
    \bR \ni t \mapsto (\bsq(t), \bsp(t)) \in \cP^S
\label{Sutherland_flow}
\ee
satisfying $(\bsq(0), \bsp(0)) = (q, p)$. Recalling 
the Sutherland Lax matrix (\ref{L}), we let $L = L(q, p)$. Keeping the 
notation (\ref{Q}) in effect, for each $t \in \bR$ we also define
\be
    \bsQ(t) 
    = \diag(\bsq_1(t), \ldots, \bsq_n(t), 
        -\bsq_1(t), \ldots, -\bsq_n(t)).
\label{bsQ}
\ee
Now, it is clear that the (complete) `geodesic' flow
\be
    \bR \ni t 
    \mapsto
    (e^Q e^{t L}, L, \xi(E)) \in \mfL_0
\label{unreduced_Sutherland_flow}
\ee
generated by the unreduced free Hamiltonian $F$ (\ref{F}) projects onto 
the (complete) reduced flow (\ref{Sutherland_flow}). In particular, utilizing
the parametrization $\Upsilon^S_0$ (\ref{Upsilon_S_0}), for all $t \in \bR$ 
we can write
\be
    e^Q e^{t L} = k_L(t) e^{\bsQ(t)} k_R(t)^{-1}
\label{G_component}
\ee
with some $k_L(t), k_R(t) \in K$. It entails the spectral identification
\be
    \sigma(e^{2 \bsQ(t)}) = \sigma(e^{2 Q} e^{t L} e^{t L^*}),
\label{Sutherland_spectral_identification_1}
\ee
which can be seen as the starting point of a purely algebraic solution 
algorithm of the Sutherland model based on matrix diagonalization.

Making use of the symplectomorphism $\cS$ (\ref{cS}), the above matrix flows
can be parametrized with the dual variables $(\lambda, \theta) = \cS(q, p)$ 
as well. First, recalling the matrices (\ref{cA}) and (\ref{h}), we introduce
the shorthand notations $\cA = \cA(\lambda, \theta)$ and $h = h(\lambda)$. 
Second, using the abbreviation defined in (\ref{Lambda}), notice that the 
parametrization $\Upsilon^R_0$ (\ref{Upsilon_R_0}) together with the defining
property of $\cS$ (\ref{cS_def}) immediately lead to the relationships
\be
    e^Q = \eta_L \cA^\half h^{-1} \eta_R^{-1}
    \quad \text{and} \quad
    L = \eta_R h \Lambda h^{-1} \eta_R^{-1}
\label{Sutherland_dual_parametrization}
\ee
with some $\eta_L, \eta_R \in K$. It readily follows that
\be
    e^{2 Q} e^{t L} e^{t L^*} 
    = \eta_R h^{-1} \cA e^{t \Lambda} h^{-2} e^{t \Lambda} h \eta_R^{-1}.
\label{Sutherland_flow_dual_parametrization}
\ee
However, since the functions (\ref{alpha&beta}) appearing in the definition 
of $h$ (\ref{h}) obey the identities
\be
    \alpha(x)^2 + \beta(x)^2 = 1,
    \quad
    \alpha(x)^2 - \beta(x)^2 = \sqrt{1 + \frac{\kappa^2}{x^2}},
    \quad
    2 \alpha(x) \beta(x) = \frac{\ri \kappa}{x},
\label{alpha&beta_identities}
\ee
we find easily that 
\be
    h^{-2} 
    = \sqrt{\bsone_N + \kappa^2 \Lambda^{-2}} 
        + \ri \kappa \bsC \Lambda^{-1}.
\label{h_formula}
\ee
Therefore, by plugging the above formulae into 
(\ref{Sutherland_spectral_identification_1}), we obtain
\be
    \sigma(e^{2 \bsQ(t)}) 
    = \sigma
        \left(
            \cA \sqrt{\bsone_N + \kappa^2 \Lambda^{-2}} e^{2 t \Lambda}
                + \ri \kappa \cA \bsC \Lambda^{-1}
        \right).
\label{Sutherland_spectral_identification_2}
\ee 
By exploiting the consequences of this formula, now we can work out the 
scattering theory of the Sutherland model. Before going into the details, 
for each $a \in \bN_n$ and $\lambda \in \mfc$ we define
\be
\begin{split}
    \Delta_a(\lambda)
    = & 
    -\half \sum_{b = 1}^{a - 1} 
        \ln \left( 1 + \frac{4 \mu^2}{(\lambda_a - \lambda_b)^2} \right)
    + \half \sum_{b = a + 1}^n 
        \ln \left( 1 + \frac{4 \mu^2}{(\lambda_a - \lambda_b)^2} \right)
    \\
    & + \half \sum_{\substack{b = 1 \\ (b \neq a)}}^n 
        \ln \left( 1 + \frac{4 \mu^2}{(\lambda_a + \lambda_b)^2} \right)
    + \half \ln \left( 1 + \frac{\nu^2}{\lambda_a^2} \right)
    + \half \ln \left( 1 + \frac{\kappa^2}{\lambda_a^2} \right).
\end{split}
\label{Delta}
\ee

\begin{LEMMA}
\label{LEMMA_Sutherland_asymptotics}
Take an arbitrary $(q, p) \in \cP^S$ and let 
$(\lambda, \theta) = \cS(q, p) \in \cP^R$. Consider the Hamiltonian flow 
of the Sutherland model
\be
    \bR \ni t \mapsto (\bsq(t), \bsp(t)) \in \cP^S
\label{LEMMA_Sutherland_flow}
\ee
satisfying $(\bsq(0), \bsp(0)) = (q, p)$. Then there are some positive 
constants $T > 0$, $r > 0$ and $C > 0$, such that for all $a \in \bN_n$ and
for all $t > T$ we have
\be
    \left\vert 
        \bsq_a(t) 
            - \left( t \lambda_a - \theta_a 
                + \half \Delta_a(\lambda) \right) 
    \right\vert \leq C e^{-t r}
    \quad \text{and} \quad
    \left\vert 
        \bsp_a(t) - \lambda_a 
    \right\vert \leq C e^{-t r}, 
\label{LEMMA_Sutherland_plus_infty}
\ee
whereas for all $a \in \bN_n$ and for all $t < - T$ we can write
\be
    \left\vert 
        \bsq_a(t) 
            - \left( t (-\lambda_a) + \theta_a 
                + \half \Delta_a(\lambda) \right) 
    \right\vert \leq C e^{t r}
    \quad \text{and} \quad
    \left\vert 
        \bsp_a(t) - (-\lambda_a) 
    \right\vert \leq C e^{t r}. 
\label{LEMMA_Sutherland_minus_infty}
\ee
\end{LEMMA}

\begin{proof}
To prove the lemma, our guiding principle is Ruijsenaars' result on the 
temporal asymptotics of the spectra of exponential matrix flows (see 
Theorem A2 in \cite{Ruij_CMP1988}). First, we introduce the $n \times n$ 
matrix $\cR_n$ with entries 
\be
    (\cR_n)_{a, b} = \delta_{a + b, n + 1}.
\label{cR}
\ee    
By conjugating with
\be
    \cW = \begin{bmatrix}
        \bsone_n & 0_n \\
        0_n & \cR_n
    \end{bmatrix} 
    \in GL(N, \bC),
\label{cW}
\ee
we also define the $N \times N$ matrices
\be
    M = \cW \cA \sqrt{\bsone_N + \kappa^2 \Lambda^{-2}} \cW^{-1},
    \quad
    D = \cW \Lambda \cW^{-1},
    \quad
    X = \ri \kappa \cW \cA \bsC \Lambda^{-1} \cW^{-1}. 
\label{M&D&X}
\ee
Let us observe that the diagonal matrix $D$ has the property that its 
eigenvalues are in strictly decreasing order on the diagonal.  

By inspecting the matrix entries of $M$ (\ref{M&D&X}), notice that for all 
$a, b \in \bN_n$ we have 
\be
    M_{a, b} = \cA_{a, b} (1 + \kappa^2 \lambda_b^{-2})^\half.
\label{M_entries}
\ee
For each $a \in \bN_n$ let $M^{(a)}$ denote the leading principal
$a \times a$ submatrix taken from the upper-left-hand corner of $M$. Since 
$M^{(a)}$ is essentially a Cauchy matrix, for its determinant we have
\be
    \det(M^{(a)})
    = \prod_{b = 1}^a 
        e^{-2 \theta_b} 
        \vert z_b(\lambda) \vert 
        (1 + \kappa^2 \lambda_b^{-2})^\half
    \prod_{1 \leq c < d \leq a} 
        (1 + 4 \mu^2 (\lambda_c - \lambda_d)^{-2})^{-1}.
\label{det_M_a}
\ee
By taking the quotients of the consecutive leading principal minors of $M$, 
we also define
\be
    m_1 = M_{1, 1}
    \quad \text{and} \quad
    m_a = \frac{\det(M^{(a)})}{\det(M^{(a - 1)})}
    \qquad (2 \leq a \leq n).
\label{m_a_def}
\ee
From equation (\ref{det_M_a}) it follows that for all $a \in \bN_n$ we 
can write
\be
    m_a 
    = e^{-2 \theta_a} 
        \vert z_a(\lambda) \vert 
        (1 + \kappa^2 \lambda_a^{-2})^\half
        \prod_{b = 1}^{a - 1} 
            (1 + 4 \mu^2 (\lambda_a - \lambda_b)^{-2})^{-1}.
\label{m_a}
\ee
Therefore, recalling the formulae (\ref{z}) and (\ref{Delta}), we end up
with the concise expression
\be
    \ln(m_a) = - 2 \theta_a + \Delta_a(\lambda).
\label{ln_m_a}
\ee 

Utilizing the above objects, we are now in a position to analyze the 
asymptotic properties of the trajectory (\ref{LEMMA_Sutherland_flow}). 
By combining equations (\ref{Sutherland_spectral_identification_2}) and 
(\ref{M&D&X}), it is clear that for all $t \in \bR$ we can write
\be
    \sigma(e^{2 \bsQ(t)}) = \sigma(M e^{2 t D} + X).
\label{Sutherland_spectral_identification_OK}
\ee
By slightly generalizing Ruijsenaars' aforementioned theorem, one can easily 
verify that the exponentially growing large eigenvalues of the matrices 
$M e^{2 t D}$ and $M e^{2 t D} + X$ have essentially the same asymptotic 
properties as $t \to \infty$. More precisely, there are some positive 
constants $T > 0$, $r > 0$ and $C > 0$, such that for all $a \in \bN_n$ 
and for all $t > T$ we have
\be
    e^{2 \bsq_a(t)} = m_a e^{2 t \lambda_a} (1 + \varepsilon_a(t)),
\label{exp_2_q_a}
\ee
where the error term $\varepsilon_a(t)$ obeys the estimation
\be
    \max \{\vert \varepsilon_a(t) \vert, 
        \vert \dot{\varepsilon}_a(t) \vert \} 
    \leq C e^{-t r} \leq \half.
\label{varepsilon_estimation}
\ee
By taking the logarithm of equation (\ref{exp_2_q_a}), it readily follows 
that
\be
    \left\vert 
        \bsq_a(t) - t \lambda_a - \half \ln(m_a) 
    \right\vert
    = \half \vert \ln(1 + \varepsilon_a(t)) \vert 
    \leq \vert \varepsilon_a(t) \vert 
    \leq C e^{-t r}.
\label{q_a_asymptotics}
\ee
Due to the formula (\ref{ln_m_a}), for the large positive values of $t$ 
the control over $\bsq_a(t)$ is complete. Next, by taking the derivative 
of equation (\ref{exp_2_q_a}), we also find
\be
    \dot{\bsq}_a(t) 
    = \lambda_a 
        + \half \frac{\dot{\varepsilon}_a(t)}{1 + \varepsilon_a(t)}.
\label{q_a_dot}
\ee
However, the Hamiltonian equations of motion yield $\bsp = \dot{\bsq}$, 
whence we obtain
\be
    \vert \bsp_a(t) - \lambda_a \vert 
    = \half \frac{\vert \dot{\varepsilon}_a(t) \vert}
        {\vert 1 + \varepsilon_a(t) \vert}
    \leq \vert \dot{\varepsilon}_a(t) \vert
    \leq C e^{-t r}.
\label{p_a_asymptotics}
\ee
    
To conclude the proof, let us observe that by conjugating with the 
$N \times N$ matrix $\cR_N$ (\ref{cR}), the spectral identification 
(\ref{Sutherland_spectral_identification_OK}) can be rewritten as
\be
    \sigma(e^{2 \bsQ(t)}) 
        = \sigma 
            \left(
                \cR_N M \cR_N^{-1} e^{2 t \cR_N D \cR_N^{-1}} 
                    + \cR_N X \cR_N^{-1}
            \right).
\label{Sutherland_alternative_spectral_identification}
\ee
Since the eigenvalues of the diagonal matrix $\cR_N D \cR_N^{-1} = - D$ are
in strictly increasing order on the diagonal, the asymptotic relationships 
(\ref{LEMMA_Sutherland_minus_infty}) for $t \to -\infty$ can be established
by the same techniques that we used for the $t \to \infty$ case.
\end{proof}

Keeping the notations introduced in lemma \ref{LEMMA_Sutherland_asymptotics},
the asymptotics can be rewritten as
\be
    \bsq_a(t) \sim t p_a^\pm + q_a^\pm
    \quad \text{and} \quad
    \bsp_a(t) \sim p_a^\pm
    \qquad
    (t \to \pm \infty), 
\label{Sutherland_asymptotics}
\ee
where the asymptotic phases and momenta have the form
\be
    q_a^\pm = \mp \theta_a + \half \Delta_a(\lambda)
    \quad \text{and} \quad
    p_a^\pm = \pm \lambda_a,
\label{Sutherland_asymptotic_phases&impulses}
\ee
respectively. Notice that the asymptotic states $(q^\pm, p^\pm)$ belong to
the manifolds
\be
    \cP^\pm 
        = \{ (x, y) = (x_1, \ldots, x_n, y_1, \ldots, y_n) \in \bR^{2 n} 
            \, | \, 
            y_1 \gtrless \ldots \gtrless y_n \gtrless 0 \},
\label{cP_pm}
\ee
that we endow with the natural symplectic forms
\be 
    \omega^\pm = \sum_{a = 1}^n \dd x_a \wedge \dd y_a.
\label{omega_pm}
\ee 
One of the principal goals of scattering theory is to study the wave maps
\be
    W^S_\pm \colon \cP^S \rightarrow \cP^\pm,
    \quad
    (q, p) \mapsto (q^\pm, p^\pm).
\label{Sutherland_wave_maps}
\ee

\begin{THEOREM}
\label{THEOREM_Sutherland_scattering}
The wave maps $W^S_\pm$ of the Sutherland model are symplectomorphisms from 
$\cP^S$ onto $\cP^\pm$. The scattering map $S^S = W^S_+ \circ (W^S_-)^{-1}$
is also a symplectomorphism of the form
\be
\cP^- \ni (x, y) \mapsto S^S(x, y) 
        = \left( -x_1 + \Delta_1(-y), \ldots, -x_n + \Delta_n(-y),
        -y_1, \ldots, -y_n \right) \in \cP^+.
\label{S_S}
\ee
\end{THEOREM}

\begin{proof}
Let us introduce the maps
\be
    \cT^S_\pm \colon \cP^R \rightarrow \cP^\pm,
    \quad
    (\lambda, \theta) 
    \mapsto 
    \left(
        \mp \theta_1 + \half \Delta_1(\lambda), 
            \ldots, \mp \theta_n + \half \Delta_n(\lambda),
        \pm \lambda_1, \ldots, \pm \lambda_n 
    \right).
\label{cT_S_pm}
\ee
Recalling $\Delta_a$ (\ref{Delta}), it is evident that $\cT^S_\pm$ are 
symplectomorphisms with inverses 
\be
    (\cT^S_\pm)^{-1} (x, y) 
    = \left(
        \pm y_1, \ldots, \pm y_n, 
            \mp x_1 \pm \half \Delta_1(\pm y), 
                \ldots, \mp x_n \pm \half \Delta_n(\pm y)
    \right).
\label{cT_pm_inverse}
\ee
Moreover, remembering the asymptotic phases and momenta 
(\ref{Sutherland_asymptotic_phases&impulses}), it readily follows that the
wave maps (\ref{Sutherland_wave_maps}) are symplectomorphisms of the form
\be
    W^S_\pm = \cT^S_\pm \circ \cS.
\label{W_S_pm}
\ee
Since $S^S = \cT^S_+ \circ (\cT^S_-)^{-1}$, the explicit formula (\ref{S_S}) 
is also immediate. 
\end{proof}

\subsection{Scattering properties of the RSvD model}
Take an arbitrary $(\lambda, \theta) \in \cP^R$ and consider the
Hamiltonian flow of the RSvD model
\be
    \bR \ni t \mapsto (\bslambda(t), \bstheta(t)) \in \cP^R
\label{RSvD_flow}
\ee
passing through the point $(\lambda, \theta)$ at $t = 0$. Remembering
the matrices (\ref{cA}), (\ref{h}), (\ref{A_bc}), and the column vector
(\ref{cV}), we introduce the abbreviations
\be
    \cA = \cA(\lambda, \theta),
    \quad
    h = h(\lambda),
    \quad
    \cA^\bc = \cA^\bc (\lambda, \theta),
    \quad
    \cV = \cV(\lambda, \theta).
\label{cA&h&A_bc&cV}
\ee
Keeping the notation displayed in equation (\ref{Lambda}), for all 
$t \in \bR$ we also the define diagonal matrix
\be
    \bsLambda(t) 
    = \diag 
        \left(
            \bslambda_1(t), \ldots, \bslambda_n(t), 
            -\bslambda_1(t), \ldots, -\bslambda_n(t)
        \right).
\label{bsLambda}
\ee
Now, bearing in mind the symplectic reduction derivation of the RSvD model, 
let us note that the (complete) `linear' flow
\be
    \bR \ni t 
    \mapsto 
    \left(
        \cA^\half h^{-1}, 
        h \Lambda h^{-1} - t \left( \cA^\bc - (\cA^\bc)^{-1} \right), 
        \xi(\cV)
    \right) 
    \in \mfL_0 
\label{RSvD_unreduced_flow}
\ee
generated by the unreduced Hamiltonian $f$ (\ref{f}) projects onto the
(complete) reduced RSvD flow (\ref{RSvD_flow}). Recalling the parametrization
$\Upsilon^R_0$ (\ref{Upsilon_R_0}), it is evident that for all $t \in \bR$
we have
\be
    h \Lambda h^{-1} - t \left( \cA^\bc - (\cA^\bc)^{-1} \right)
    = k_R(t) h(\bslambda(t)) \bsLambda(t) h(\bslambda(t))^{-1} k_R(t)^{-1}
\label{mfg_component}
\ee
with some $k_R(t) \in K$. In particular, we obtain the spectral identification
\be
    \sigma(\bsLambda(t)) 
    = \sigma
        \left(
            h \Lambda h^{-1} - t \left( \cA^\bc - (\cA^\bc)^{-1} \right)
        \right).
\label{RSvD_spectral_identification}
\ee
The point is that each trajectory of the RSvD dynamics can be recovered by 
simply diagonalizing a linear matrix flow.

Utilizing the duality symplectomorphism $\cS$ (\ref{cS}), our next
goal is to parametrize the above linear matrix flow with the dual variables
$(q, p) = \cS^{-1}(\lambda, \theta) \in \cP^S$. Recalling the Sutherland
Lax matrix (\ref{L}), we first define $L = L(q, p)$. At this point, 
remembering the notation (\ref{Q}), it is clear that the form of 
$\Upsilon^S_0$ (\ref{Upsilon_S_0}) and the defining property of $\cS$ 
(\ref{cS_def}) lead to the relationships
\be
    \cA^\half h^{-1} = \eta_L e^Q \eta_R^{-1}
    \quad \text{and} \quad
    h \Lambda h^{-1} = \eta_R L \eta_R^{-1}
\label{RSvD_dual_parametrization}
\ee 
with some $\eta_L, \eta_R \in K$. It entails 
$\cA^\bc = \eta_R e^{2 Q} \eta_R^{-1}$, therefore the matrix
\be
    \cA^\bc - (\cA^\bc)^{-1} = 2 \eta_R \sinh(2 Q) \eta_R^{-1}
\label{RSvD_slope} 
\ee
has a simple spectrum. Moreover, recalling the relationship 
(\ref{RSvD_spectral_identification}), for the spectrum of $\bsLambda(t)$
we obtain the particularly useful formula
\be
    \sigma(\bsLambda(t)) 
    = \sigma 
        \left( 
            L - 2 t \sinh(2 Q)
        \right).
\label{RSvD_spectral_identification_2}
\ee
\begin{LEMMA}
\label{LEMMA_RSvD_asymptotics}
Take an arbitrary point $(\lambda, \theta) \in \cP^R$ and let 
$(q, p) = \cS^{-1}(\lambda, \theta) \in \cP^S$. Consider the Hamiltonian 
flow of the RSvD model
\be
    \bR \ni t \mapsto (\bslambda(t), \bstheta(t)) \in \cP^R
\label{LEMMA_RSvD_flow}
\ee
satisfying $(\bslambda(0), \bstheta(0)) = (\lambda, \theta)$. Then there 
are some positive constants $T > 0$ and $C > 0$, such that for all 
$a \in \bN_n$ and for all $t > T$ we have
\be
    \left\vert 
        \bslambda_a(t) 
            - \left( 2 t \sinh(2 q_a) - p_a \right) 
    \right\vert \leq C t^{-1}
    \quad \text{and} \quad
    \left\vert 
        \bstheta_a(t) - q_a 
    \right\vert \leq C t^{-2}, 
\label{LEMMA_RSvD_plus_infty}
\ee
whereas for all $a \in \bN_n$ and for all $t < - T$ we can write
\be
    \left\vert 
        \bslambda_a(t) 
            - \left( 2 t \sinh(-2 q_a) + p_a \right) 
    \right\vert \leq C \vert t \vert^{-1}
    \quad \text{and} \quad
    \left\vert 
        \bstheta_a(t) - (-q_a) 
    \right\vert \leq C \vert t \vert^{-2}. 
\label{LEMMA_RSvD_minus_infty}
\ee
\end{LEMMA}

\begin{proof}
Making use of $\cR_n$ (\ref{cR}), we define 
\be
    \cW = \begin{bmatrix}
        0_n & \bsone_n \\
        \cR_n & 0_n
    \end{bmatrix}
    \in GL(N, \bC).
\label{RSvD_cW}
\ee 
With the aid of the $N \times N$ matrices 
\be
    D = -2 \cW \sinh(2 Q) \cW^{-1} 
    \quad \text{and} \quad
    M = \cW L \cW^{-1}, 
\label{RSvD_D&M}
\ee
the spectral identification (\ref{RSvD_spectral_identification_2}) can be 
cast into the form
\be
    \sigma(\bsLambda(t)) = \sigma(t D + M),
\label{RSvD_spectral_identification_OK}
\ee
where $D$ is a diagonal matrix with its eigenvalues in strictly decreasing
order on the diagonal. Therefore elementary perturbation theory can be 
readily applied to analyze the properties of $\bsLambda(t)$ as 
$t \to \infty$. (For a short account on the relevant facts from
perturbation theory see e.g. Theorem A1 in \cite{Ruij_CMP1988}.)
More precisely, there are some positive constants $T_0 > 0$ 
and $\cC > 0$, such that for all $a \in \bN_n$ and for all $t > T_0$ the 
eigenvalue $\bslambda_a(t)$ has the form
\be
    \bslambda_a(t) 
    = t D_{a, a} + M_{a, a} + \varepsilon_a(t)
    = 2 t \sinh(2 q_a) - p_a + \varepsilon_a(t),
\label{bslambda_a}
\ee
where the error term $\varepsilon_a(t)$ satisfies
\be 
    \vert \varepsilon_a(t) \vert \leq \cC t^{-1} 
    \quad \text{and} \quad    
    \vert \dot{\varepsilon}_a(t) \vert \leq \cC t^{-2}.
\label{RSvD_error_term_estimation}
\ee

Since the asymptotic property of $\bslambda(t)$ is under control, 
now we can turn our attention to the asymptotic analysis of $\bstheta(t)$. 
Making use of the equations of motion generated by the RSvD Hamiltonian $H^R$ 
(\ref{H_R}), we find
\be
    \dot{\bslambda}_a(t) = 2 \sinh(2 \bstheta_a(t)) v_a(\bslambda(t)).
\label{bslambda_dot}
\ee
From the asymptotic behavior of $\bslambda(t)$ (\ref{bslambda_a})
it is clear that there are some constants $T_1 \geq T_0$ and $\cK > 0$,
such that for all $a, b, c \in \bN_n$, $a \neq b$, and for all
$t > T_1$ we have
\be
    \vert \bslambda_a(t) - \bslambda_b(t) \vert \geq \cK t
    \quad \text{and} \quad
    \bslambda_c(t) \geq \cK t.
\label{cK_estimation}
\ee
By inspecting $v_a$ (\ref{v_a}), it is also evident that there are some 
constants $T \geq T_1$ and $\cH > 0$, such that for all $a \in \bN_n$ 
and for all $t > T$ we can write
\be
    v_a(\bslambda(t)) \leq 1 + \cH t^{-2}.
\label{cH_estimation}
\ee
By combining equations (\ref{bslambda_a}) and (\ref{bslambda_dot}),
for all $a \in \bN_n$ and for all $t > T$ we obtain
\be
    \sinh(2 \bstheta_a(t)) - \sinh(2 q_a) 
    = \frac{
        \left( 1 - v_a(\bslambda(t)) \right) \sinh(2 q_a)  
            + \half \dot{\varepsilon}_a(t)}
        {v_a(\bslambda(t))},
\label{RSvD_bstheta&q}
\ee
whence the estimation
\be
    \vert \bstheta_a(t) - q_a \vert
    \leq 
    \frac{\vert \sinh(2 \bstheta_a(t)) - \sinh(2 q_a) \vert}{2}
    \leq 
    \left( 
        \frac{\cH \sinh(2 q_1)}{2}  + \frac{\cC}{4} 
    \right) \frac{1}{t^2}
\label{bstheta_a}
\ee
is immediate. Thus the asymptotic relationships 
(\ref{LEMMA_RSvD_plus_infty}) are established.

To conclude, let us note that by applying the same techniques on the matrix 
flow 
\be
    t \mapsto \cR_N (t D + M) \cR_N^{-1},
\label{RSvD_flow_minus_infty}
\ee 
one can easily prove the remaining relationships 
(\ref{LEMMA_RSvD_minus_infty}) as well. 
\end{proof}

Keeping the notations of lemma \ref{LEMMA_RSvD_asymptotics}, our results
on the asymptotics can be rephrased as
\be
    \bslambda_a(t) \sim 2 t \sinh(2 \theta_a^\pm) + \lambda_a^\pm
    \quad \text{and} \quad
    \bstheta_a(t) \sim \theta_a^\pm
    \qquad
    (t \to \pm \infty),
\label{RSvD_asymptotics}
\ee
where for all $a \in \bN_n$ we have
\be 
    \lambda_a^\pm = \mp p_a 
    \quad \text{and} \quad 
    \theta_a^\pm = \pm q_a.
\label{RSvD_asymptotic_phases_and_momenta}
\ee
Notice that the asymptotic states $(\lambda^\pm, \theta^\pm)$ belong to 
the phase spaces $\cP^\pm$ (\ref{cP_pm}), therefore the RSvD model can be 
characterized by the wave maps
\be
    W^R_\pm \colon \cP^R \rightarrow \cP^\pm,
    \quad
    (\lambda, \theta) \mapsto (\lambda^\pm, \theta^\pm).
\label{RSvD_wave_maps}
\ee

\begin{THEOREM}
\label{THEOREM_RSvD_scattering}
The wave maps $W^R_\pm$ of the RSvD model are symplectomorphisms from 
$\cP^R$ onto $\cP^\pm$. The scattering map is also a symplectomorphism 
of the form
\be
    S^R = W^R_+ \circ (W^R_-)^{-1} \colon \cP^- \rightarrow \cP^+,
    \quad
    (x, y) \mapsto S^R(x, y) = (-x, -y).
\label{S_R}
\ee
\end{THEOREM}

\begin{proof}
With the aid of the symplectomorphisms
\be
    \cT^R_\pm \colon \cP^S \rightarrow \cP^\pm,
    \quad
    (q, p) \mapsto (\mp p, \pm q)
\label{cT_R_pm}
\ee
the wave maps (\ref{RSvD_wave_maps}) can be realized as compositions 
of symplectomorphisms of the form
\be
    W^R_\pm = \cT^R_\pm \circ \cS^{-1}.
\label{W_R_pm}
\ee
Due to the relationship $S^R = \cT^R_+ \circ (\cT^R_-)^{-1}$, the formula 
(\ref{S_R}) also follows.
\end{proof}

\section{Discussion}
\label{SECTION_Discussion}
\setcounter{equation}{0}
For many physical systems it is the scattering theory that provides the main 
tool to investigate the properties of the constituent particles and the 
nature of their interaction. At the same time, scattering theory is a 
notoriously difficult subject heavily relying on non-trivial techniques from 
hard mathematical analysis. However, by exploiting the action-angle duality
between the Sutherland and the RSvD models, a careful analysis of their 
algebraic solution algorithms led us to rigorous temporal asymptotics of the
trajectories. Having a glance at the structure of the resulting wave maps 
$W^S_\pm$ (\ref{W_S_pm}) and $W^R_\pm$ (\ref{W_R_pm}), it is transparent 
that they are made of two strikingly different building blocks. The 
explicitly defined maps $\cT^S_\pm$ (\ref{cT_S_pm}) and $\cT^R_\pm$ 
(\ref{cT_R_pm}) are manifestly symplectic, meanwhile the construction and 
the symplecticity of $\cS$ (\ref{cS}) hinges on the symplectic reduction 
derivation of the Sutherland and the RSvD systems. That is, beside providing 
canonical action-angle variables, the geometric ideas culminating in the 
duality symplectomorphism $\cS$ (\ref{cS}) prove to be fundamental in the 
scattering theory as well.

Looking at the scattering maps $S^S$ (\ref{S_S}) and $S^R$ (\ref{S_R}),
we see that, up to an overall sign, the asymptotic momenta of both the 
hyperbolic $BC_n$ Sutherland and the rational $BC_n$ RSvD models are 
preserved. Following Ruijsenaars' terminology 
\cite{Ruij_FiniteDimSolitonSystems}, we can  say that these 
integrable systems are \emph{pure soliton systems} of type $BC_n$.
Moreover, for the Sutherland model the classical phase shifts are completely 
determined by the $2$-particle processes and the $1$-particle scattering on 
the external field. As for the rational RSvD system, the phase shifts are 
trivial. In other words, in both cases the scattering map has a factorized 
form. Due to the puzzling connection with the theory of solitons, it would 
be of considerable interest to develop an analogous theory at the quantum 
level, too. Relatedly, we find it a much more ambitious, but equally 
motivated research problem to classify the pure soliton systems associated 
with arbitrary root systems. We wish to come back to these issues in later 
publications.   

\medskip
\noindent
\textbf{Acknowledgments.}
This work was partially supported by the Hungarian Scientific Research 
Fund (OTKA) under Grant No. K 77400. The work was also supported by the 
European Union and co-funded by the European Social Fund under the project 
``Telemedicine-focused research activities on the field of Mathematics,
Informatics and Medical Sciences" of project number 
``T\'AMOP-4.2.2.A-11/1/KONV-2012-0073".



\begin{thebibliography}{99}

    \bibitem{Calogero}
		F.~Calogero, Solution of the one-dimensional $N$-body problem with 
		quadratic and/or inversely quadratic pair potentials, 
		J. Math. Phys. 12 (1971) 419-436.
	
	\bibitem{Sutherland}
		B.~Sutherland, 
		Exact results for a quantum many body problem in one dimension, 
		Phys. Rev. A 4 (1971) 2019-2021.
	
	\bibitem{Moser_1975}
		J.~Moser,
		Three integrable Hamiltonian systems connected with isospectral 
		deformations, 
		Adv. Math. 16 (1975) 197-220.
		
	\bibitem{Olsha_Pere_76}
        M.A.~Olshanetsky and A.M.~Perelomov,
        Completely integrable Hamiltonian systems connected with semisimple 
        Lie algebras,
        Invent. Math. 37 (1976) 93-108.
		
	\bibitem{Ruij_Schneider}
		S.N.M.~Ruijsenaars and H.~Schneider,
		A new class of integrable models and its relation to solitons,
		Ann. Phys. (N.Y.) 170 (1986) 370-405.
	
	\bibitem{van_Diejen}
		J.F.~van Diejen, 
		Deformations of Calogero--Moser systems and finite Toda chains, 
		Theor. Math. Phys. 99 (1994) 549-554.
    
    \bibitem{Ruij_CMP1988}
        S.N.M.~Ruijsenaars, 
        Action-angle maps and scattering theory for some finite dimensional 
        integrable systems I. The pure soliton case, 
        Commun. Math. Phys. 115 (1988) 127-165.
        
    \bibitem{Ruij_RIMS_2}
        S.N.M.~Ruijsenaars, 
        Action-angle maps and scattering theory for some finite dimensional 
        integrable systems II. Solitons, antisolitons and their bound states,
        Publ. RIMS 30 (1994) 865-1008.

    \bibitem{Ruij_RIMS_3}
        S.N.M.~Ruijsenaars, 
        Action-angle maps and scattering theory for some finite dimensional 
        integrable systems III. Sutherland type systems and their duals,
        Publ. RIMS 31 (1995) 247-353.

    \bibitem{Babelon_Bernard}
		O.~Babelon and D.~Bernard,
		The sine-Gordon solitons as a $N$-body problem,
		Phys. Lett. B 317 (1993) 363-368.

	\bibitem{Kulish_1976}
		P.P.~Kulish,
		Factorization of the classical and the quantum $S$ matrix and 
		conservation laws,
		Theor. Math. Phys. 26 (1976) 132-137.
		
	\bibitem{Moser_1977}
		J.~Moser,
		The scattering problem for some particle systems on the line, 
		in: Lecture Notes in Mathematics 597, Springer, 1977, pp. 441-463. 	

    \bibitem{Ruij_FiniteDimSolitonSystems}
        S.N.M.~Ruijsenaars,
        Finite-dimensional soliton systems, 
        in: B. Kupershmidt (Ed.), Integrable and superintegrable systems, 
        World Scientific, 1990, pp. 165-206.

    \bibitem{Kapustin_Skorik}
        A.~Kapustin and S.~Skorik,
        On the non-relativistic limit of the quantum sine-Gordon model with 
        integrable boundary condition,
        Phys. Lett. A 196 (1994) 47-51.

	\bibitem{Pusztai_JPA2011}
		B.G.~Pusztai, 
		On the scattering theory of the classical hyperbolic $C_n$ 
		Sutherland model,
		J. Phys. A 44 (2011) 155306.

    \bibitem{Pusztai_NPB2012}
        B.G.~Pusztai,
        The hyperbolic $BC_n$ Sutherland and the rational $BC_n$ 
        Ruijsenaars--Schnei\-der--van Diejen models: Lax matrices and duality,
        Nucl. Phys. B 856 (2012) 528-551.

    \bibitem{Pusztai_NPB2011}
        B.G.~Pusztai,
        Action-angle duality between the $C_n$-type hyperbolic Sutherland 
        and the rational Ruijsenaars--Schneider--van Diejen models,
        Nucl. Phys. B 853 (2011) 139-173.

\end{thebibliography}
\end{document}